%% file: main.tex
\documentclass{IEEEtran}
\usepackage{cite}
\usepackage{amsmath,amssymb,amsfonts}
\usepackage{graphicx}
\usepackage{textcomp}

\usepackage{ifpdf}
\usepackage{cite, graphicx, graphics}

\usepackage{amsmath}
\usepackage{algorithm}
\usepackage{algpseudocode}
\usepackage{array}

\usepackage{fixltx2e}
\usepackage{stfloats}

\ifCLASSOPTIONcaptionsoff
  \usepackage[nomarkers]{endfloat}
\fi

\usepackage{xcolor}
\iftrue
\newcommand{\berivan}[1]{{\color{red}\textbf{BI}: #1}}
\fi
\usepackage{url}
\usepackage{subfigure}
\usepackage{amsthm}
\newtheorem{theorem}{\bf Theorem}[section]

\newtheorem{corollary}{\bf Corollary}[section]

\newtheorem{defn}{Definition}

\newcommand{\be}{\begin{equation}}
\newcommand{\ee}{\end{equation}}

\def \bX {\mathbf{X}}

\def \bN {\mathbf{N}}
\def \hX {\hat{X}}
\def \bhX {\mathbf{\hat{X}}}

\def \bZ {\mathbf{Z}}

\newcommand{\E}{\mathbb{E}} 

\def\BibTeX{{\rm B\kern-.05em{\sc i\kern-.025em b}\kern-.08em
    T\kern-.1667em\lower.7ex\hbox{E}\kern-.125emX}}
\begin{document}
\title{Lossy Compression of Noisy Data for \\ Private and Data-Efficient Learning}
\author{Berivan Isik and~Tsachy Weissman,~\IEEEmembership{Fellow,~IEEE}
\thanks{B. Isik and T. Weissman are with the Department of Electrical Engineering, Stanford University, Stanford, CA 94305, USA
(e-mail: berivan.isik@stanford.edu; tsachy@stanford.edu).}
\thanks{Published at the IEEE Journal on Selected Areas in Information Theory (JSAIT). Preliminary version \cite{isik2022learning} was presented at the IEEE International Symposium on Information Theory (ISIT), 2022.}
}

\maketitle

\input{sections/00-abstract}


\input{sections/01-introduction}
\input{sections/02-preliminary}
\input{sections/03-method}
\input{sections/04-experiments}
\input{sections/05-conclusion}
\input{sections/06-acknowledgement}

\bibliographystyle{IEEEtran}
\bibliography{refs}

\end{document}

%% file: sections/00-abstract.tex
\begin{abstract}
Storage-efficient  privacy-preserving learning is crucial due to increasing amounts of sensitive user data required for modern learning tasks. We propose a framework for reducing the storage cost of user data while at the same time providing privacy guarantees, without essential loss in the utility of the data for learning. Our method comprises noise injection followed by lossy compression. We show that, when appropriately matching the lossy compression to the distribution of the added noise, the compressed examples converge, in distribution, to that of the noise-free training data as the sample size of the training data (or the dimension of the training data) increases. In this sense, the utility of the data for learning is essentially maintained, while reducing storage and privacy leakage by quantifiable amounts. We present experimental results on the CelebA dataset for gender classification and find that our suggested pipeline delivers in practice on the promise of the theory: the individuals in the images are unrecognizable (or less recognizable, depending on the noise level), overall storage of the data is substantially reduced, with no essential loss (and in some cases a slight boost) to the classification accuracy. As an added bonus, our experiments suggest that our method yields a substantial boost to robustness in the face of adversarial test data.     
\end{abstract}

%% file: sections/01-introduction.tex
\section{Introduction}
\label{intro}
One of the most crucial factors contributing to the recent success of machine learning is the wide availability  
of user data \cite{jordan2015machine}. However, relying on such data brings several challenges in storage and user privacy. While privacy-preserving methods for machine learning have been studied extensively, efficient storage of data for learning (a major problem even for synthetic datasets such as ImageNet~\cite{krizhevsky2012imagenet} and CelebA~\cite{celeba}) remains largely unexplored. In this work, we propose a framework to tackle the two problems jointly. 
 We seek to develop a storage-efficient privacy-guaranteeing processing procedure that preserves the utility of the data for learning.  

\begin{figure}[h]
\centering
\includegraphics[width=\columnwidth]{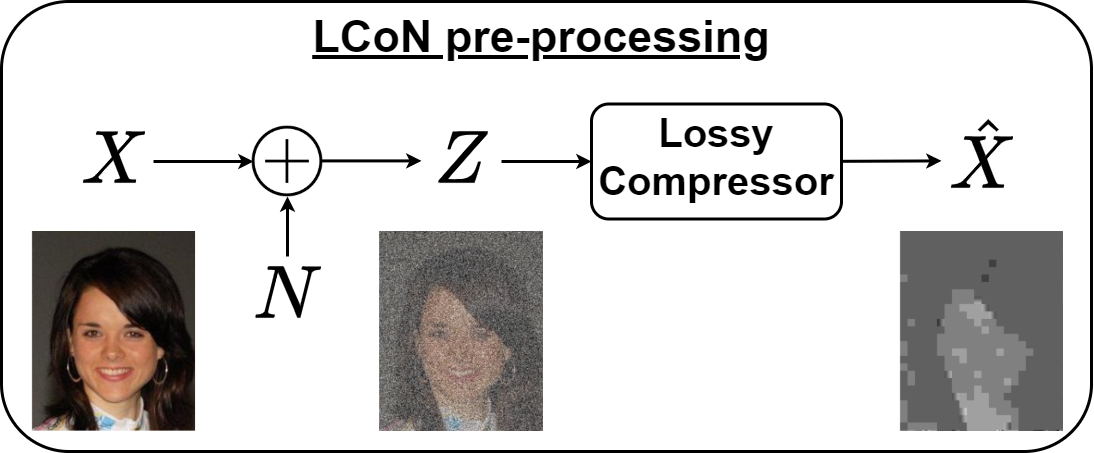}
\caption{Proposed data pre-processing framework. $X$: noise-free data, $N$: added noise, $Z$: noisy data, $\hat{X}$: reconstructions from lossy compression of the noisy data. $\hat{X}$ are then used for the learning, in lieu of $X$. A sample of noise-free, noise-injected, and lossily compressed images from the CelebA dataset are given at the bottom. Here JPEG compression with quality factor 1 is applied to (Gaussian) noisy images with $12$ dB PSNR.}
\label{fig:framework}
\end{figure}

To achieve this goal, we  first inject noise $\bN$ to the (learning data) examples $\bX$ and then lossily compress  the noisy examples $\bZ$ (see Fig.~\ref{fig:framework}). The reconstructions from the lossy compression of noisy (LCoN) examples $\bhX$ are then used for the learning.  
The lossy compression is done under a distortion criterion and level that are matched to the noise characteristics in a way we prescribe below.  
For data efficiency, we aim to achieve a compression rate close to the optimum, as characterized by the rate-distortion function associated with the noisy data. As for privacy, following \cite{makhdoumi2014information, sankar2013utility, du2012privacy, makhdoumi2013privacy, chatzikokolakis2010statistical, rebollo2009t, huang2018generative} and references therein, we guarantee an upper bound on the privacy leakage as measured by mutual information between the original data $\bX$ and that retained $\bhX$. \footnote{As is exemplified in Fig.~\ref{fig:framework}, the lossy compression step may result in a substantial further reduction of  the privacy leakage over merely noise corrupting the data.}  Our framework naturally preserves differential privacy as well \cite{dwork2006our}.   

We show that 
this procedure 
achieves our goal, which   might seem surprising at first glance in light of results from the literature on privacy and robustness  showing significant degradation in performance of the trained model when data are corrupted by noise \cite{mcpherson2016defeating, 43405}. Nevertheless, in a sense we make precise, this problem is alleviated in our framework due to the effective denoising that occurs when noisy data are lossily compressed. More concretely, when the distortion criterion and level in the lossy compression are matched to the noise characteristics, the lossily compressed noisy data samples  converge, in distribution, to that of the noise-free data  because, in effect,  they are samples from the posterior distribution of the noise-free data given its noise-corrupted versions. 
The learning then is performed on data with the ``right'' statistics so, in principle, should entail no performance loss in the downstream inference tasks.

Our initial experimentation with gender classification on the CelebA dataset seems in agreement with the theory. 
For example,  one working point of our method decreases the cost of storing the data (in bits) by a factor of two, provides privacy guarantees by adding Gaussian noise (with varying variance where the individuals in the noisy images were unrecognized by the authors), while achieving better accuracy than the benchmark methods.  
Furthermore, our method yields substantial performance boosts over the benchmark methods when tested on  
 adversarially generated data. 

Our main contributions can be summarized as: 

\begin{enumerate}
\item We propose a framework for data-efficient privacy-preserving pre-processing that retains the utility/quality of the data for learning by essentially preserving its distributional properties. We call it LCoN pre-processing since it contains \textbf{L}ossy \textbf{Co}mpression of \textbf{N}oisy data. 

\item We present initial experimentation demonstrating the efficacy of our suggested pre-processing pipeline on the CelebA dataset not only with respect to the criteria that motivated its design, but also in providing robustness to adversarial data. 
\end{enumerate}
We provide a brief summary of the related work in Section~\ref{related}, introduce some notation and basic concepts of relevance in Section~\ref{preliminary}, recall a key result about the precise sense and conditions under which reconstructions from lossy compression of noisy data are samples from the posterior distribution of the noise-free given the noisy data in Section~\ref{samples_posterior}, build on this relation to guide and justify the construction of our proposed data pre-processing framework in Section~\ref{application}, present the experimental results in Section~\ref{experiments}, and conclude in Section~\ref{conclusion}.

%% file: sections/02-preliminary.tex
\section{Related Work}
\label{related}
In this section, we briefly summarize the literature on empirical distribution of good codes, privacy, and robustness as they are related to our approach. 

\subsection{Empirical Distribution of Good Codes}
Following the analogous findings for good channel codes (a code approaching capacity with vanishing probability of error) \cite{shamai1997empirical}, \cite{ordentlichweissman} proved a similar result for good rate-constrained source codes \cite{ordentlichweissman}. In particular, they proved 
that the empirical distribution of any good rate-constrained source code approaches to the joint distribution attaining a point in the rate-distortion curve whenever this joint distribution is unique. In the same work, they also proved the denoising property of good lossy compressors, which will be covered in more detail in later sections and will be the basis of our work.      

\subsection{Privacy}
There are a number of existing information-theoretic tools to measure privacy such as mutual information \cite{zhu2005anonymity, chatzikokolakis2008anonymity, chatzikokolakis2010statistical} and rate-distortion theory \cite{moraffah2015information, basciftci2016privacy, asoodeh2014notes, asoodeh2016information, bonomi2016information}. We use the mutual information as a measure of privacy leakage as proposed in \cite{makhdoumi2013privacy, makhdoumi2014information, liao2018tunable} to formulate privacy-utility trade-offs. 
Our work adopts the inference threat model introduced in  \cite{du2012privacy}, where a user has a private data $S \in \mathcal{S}$ correlated with $X \in \mathcal{X}$ and releases a distorted version of $X$ denoted as $\hat{X} \in \mathcal{\hat{X}}$ while an adversary selects a distribution $q$ from $\mathcal{P}_S$ that minimizes an inference cost function $C(S,q)$. This threat model was also adapted by \cite{makhdoumi2014information}, where the authors studied the privacy metric under log-loss (self-information) cost function $C(S,q) = - \log{q(s)}$ and showed that privacy leakage can be measured by the mutual information $I(S; \hat{X})$. The authors of \cite{makhdoumi2014information} further showed that while $I(S;\hat{X})$ is the exact privacy leakage under log-loss cost function, any bounded cost function can be upper bounded by a constant factor of $\sqrt{I(S;\hat{X})}$, indicating that minimizing $I(S;\hat{X})$ is a judicious goal toward mitigating  privacy leakage under any bounded cost function. We borrow the privacy measure from this work as the mutual information between the private and the released data and denote it as $I(X;\hat{X})$, taking $X=S$. 

Different from the mutual information privacy, differential privacy focuses on the problem of learning aggregate statistics by collecting data from \emph{several} users \cite{dwork2006differential, dwork2009differential}. In Section~\ref{LCON_privacy}, we briefly touch on differential privacy and show that our framework naturally provides local differential privacy guarantees in the case of specific noise distributions such as Gaussian and Laplacian.

\subsection{Robustness}
There has been significant interest in enhancing robustness of deep neural networks as they are known to be vulnerable against adversarial examples \cite{43405}. 
One robustness strategy that is most related to our work is image compression as a data pre-processing step \cite{ robustness1, robustness_jpeg1, robustness_jpeg3, robustness_Jpeg4, robustness_jpeg5}. The authors of \cite{robustness_jpeg1} studied the impact of JPEG compression on robustness against adversarial examples and found out that, under moderate compression rates, JPEG compression may enhance robustness. While their results were mostly empirical, their work accelerated research on modifying existing image compression methods to specifically eliminate adversarial effects \cite{robustness_Jpeg4}. 
Another related approach to robustness is injecting Gaussian noise to the adversarial data \cite{45818}. We note that our framework does not explicitly target robustness (we do not yet provide any theoretical guarantees for robustness) but it seems to naturally enhance it due to the noise injection and lossy compression steps (see Fig.~\ref{fig:framework}). Our experimental results confirm that our method also boosts robustness without essential loss (decrease) of utility (accuracy).    

\section{Preliminaries}
\label{preliminary}

\subsection{$k$-th order distribution induced by $P_{X^n}$}
\label{subsection kth order distribution induced}
We provide the definitions of $Q_{X}^{\sf{ave}, (n)}$ and $Q_{X^k}^{\sf{ave}, (n)}$ as follows.  

 \begin{defn}[$Q_{X}^{\sf{ave}, (n)}$] Consider a random $n$-tuple $X^n$. For $J \sim \mbox{Unif} \{ 1, 2, \ldots, n \} $ and independent of $X^n$, $Q_{X}^{\sf{ave}, (n)}$ is the law of $X_J$.
 \end{defn}
In other words, $Q_{X}^{\sf{ave}, (n)}$ denotes the distribution of the random variable obtained by choosing one of the $n$ components of $X^n$ at random, or the ``average''  of the marginal laws $\{ P_{X_i} \}_{i=1}^n$, and hence the superscript.  
  We write $Q_{X}^{\sf{ave}, (n)} [P_{X^n}]$ when we want to make its dependence on the law of $X^n$ explicit.

  \begin{defn}[$Q_{X^k}^{\sf{ave}, (n)}$] For $k \leq n$, and $J \sim \mbox{Unif} \{ 1, 2, \ldots, n - k + 1 \} $ independent of $X^n$, $Q_{X^k}^{\sf{ave}, (n)}$ denotes the law of the $k$-tuple $X_{J}^{J+k-1}$.
     
 \end{defn}
 
 In other words, $Q_{X^k}^{\sf{ave}, (n)}$ is the law  obtained by averaging the marginal $k$-tuple laws $\left\{ P_{X_i^{i+k-1}} \right\}_{i=1}^{n - k + 1}$. 
    We write $Q_{X^k}^{\sf{ave}, (n)} [P_{X^n}]$ when we want to make its dependence on $P_{X^n}$ explicit. We extend this notation in the obvious way to $Q_{X,Y}^{\sf{ave}, (n)} = Q_{X,Y}^{\sf{ave}, (n)} [ P_{X^n, Y^n} ] $ and 
  $Q_{X^k,Y^k}^{\sf{ave}, (n)} = Q_{X^k,Y^k}^{\sf{ave}, (n)} [ P_{X^n, Y^n} ] $.

\subsection{$k$-th order empirical distribution induced by $x^n$}
\label{subsection $k$-th order empirical distribution induced by $x^n$}
Now, we provide the definitions of the empirical distributions $Q_{X}^{\sf{emp}, (n)}$ and $Q_{X^k}^{\sf{emp}, (n)}$.

 \begin{defn}[$Q_{X}^{\sf{emp}, (n)}$] For a fixed finite-alphabet $n$-tuple $x^n$, $Q_{X}^{\sf{emp}, (n)} [ x^n]$ is a probability mass function (PMF) on the finite alphabet $\mathcal{X}$ in which the components of $x^n$ reside, with  
$Q_{X}^{\sf{emp}, (n)} [x^n](a)$ denoting the probability it assigns to $a \in \mathcal{X}$, namely the fraction of times the symbol $a$ appears along the $n$-tuple $x^n$. 
 \end{defn}
 
In words, $Q_{X}^{\sf{emp}, (n)}[x^n]$ denotes the empirical (first-order) distribution that $x^n$ induces. To simplify the notation, we suppress the dependence on $x^n$, using $Q_{X}^{\sf{emp}, (n)}$ when $x^n$ should be clear from the context.  

 \begin{defn}[$Q_{X^k}^{\sf{emp}, (n)}$] For $k \leq n$, $Q_{X^k}^{\sf{emp}, (n)} [ x^n]$ is a PMF of a $k$-tuple, with  
$Q_{X^k}^{\sf{emp}, (n)} [x^n](a^k)$ denoting the probability it assigns to $a^k \in \mathcal{X}^k$,  the fraction of times the $k$-tuple $a^k$ appears along the $n$-tuple $x^n$.
 \end{defn}
  
Equivalently, $Q_{X^k}^{\sf{emp}, (n)} [x^n]$ denotes the empirical distribution of $k$-tuples along $x^n$.  Here too we suppress the dependence on $x^n$ and write $Q_{X^k}^{\sf{emp}, (n)}$ when $x^n$ should be clear from the context. We extend this notation to $Q_{X,Y}^{\sf{emp}, (n)} = Q_{X,Y}^{\sf{emp}, (n)} [ x^n, y^n ]$ and 
  $Q_{X^k,Y^k}^{\sf{emp}, (n)} = Q_{X^k,Y^k}^{\sf{emp}, (n)} [ x^n, y^n ]$  in the obvious ways.

\subsection{Relationship between $Q_{X}^{\sf{ave}, (n)}$ and $Q_{X}^{\sf{emp}, (n)}$}
When $X^n$ is stochastic, so is  $Q_{X^k}^{\sf{emp}, (n)} = Q_{X^k}^{\sf{emp}, (n)} [X^n]$, and for any $a^k \in \mathcal{X}^k$, we have  
\be 
\label{eq: exp of emp dist is kth order dist}
\E \left[ Q_{X^k}^{\sf{emp}, (n)} (a^k) \right]  = Q_{X^k}^{\sf{ave}, (n)} (a^k) . 
\ee 
Note further that, letting $\stackrel{n \rightarrow \infty}{\Longrightarrow}$ denote convergence in distribution, in any scenario where $ Q_{X^k}^{\sf{emp}, (n)} \stackrel{n \rightarrow \infty}{\Longrightarrow}  \mu_{X^k} \ \ \ a.s. $ 
for some PMF on $k$-tuples  $\mu_{X^k}$, we also have, by (\ref{eq: exp of emp dist is kth order dist}) and the bounded convergence theorem, $ Q_{X^k}^{\sf{ave}, (n)} \stackrel{n \rightarrow \infty}{\Longrightarrow}  \mu_{X^k} .$ 
Thus, convergence of $Q_{X^k}^{\sf{emp}, (n)}$ is stronger than (implies) convergence of $Q_{X^k}^{\sf{ave}, (n)}$. 


%% file: sections/03-method.tex
\section{Samples from the Posterior via Noisy Lossy Compression}
\label{samples_posterior}
 Consider the canonical setting where  the components of the noise-free $\bX$, noisy $\bZ$, and reconstructed sources $\bhX$ in Fig.~\ref{fig:framework} all take values in the same finite $Q$-ary alphabet $\mathcal{A} = \{0,1, \ldots, Q-1\}$.  The noise-free source $\mathbf{X} = (X_1, X_2, \ldots )$ is stationary ergodic and corrupted by additive memoryless  noise $\bN$. That is, we assume the components of the noisy observation process  
$\mathbf{Z}$ are given by 
\be \label{eq: noise model} 
 Z_i = X_i + N_i ,  
 \ee 
where the $N_i$s are IID$\sim N$, independent (collectively) of $\mathbf{X}$, and addition in (\ref{eq: noise model}) is in the mod-$Q$ sense\footnote{The framework and results have natural analog analogues, where the alphabet can be the real line or any Euclidean space and addition is in the usual sense. We assume here the finite alphabet setting for concreteness, for avoiding unnecessary technicalities, and because it is better connected to practice where the alphabets are ultimately finite.}.  We assume the distribution of the noise to be ``non-singular'' in the sense that the Toeplitz matrix whose rows are shifted versions of the row vector representing the PMF of $N$ is invertible, a benign condition
guaranteeing a one-to-one correspondence between the distributions of the noise-free and noisy sources \cite{dudepaper}.     We construct a difference distortion measure $\rho_N: \mathcal{A} \rightarrow [0, \infty]$ from the distribution of the noise according to 
\be 
\label{eq: distortion induced by noise}
\rho_N (a) = \log \frac{1}{\Pr (N=a)} . \ee 
This construction is such that the distribution of the noise has the maximum entropy with respect to $\rho_N$. That is, defining the max-entropy function induced by $\rho_N$ as 
\be 
\phi_N (D) = \max \{  H(\tilde{N}) : E \rho_N ( \tilde{N}) \leq D \} , 
\ee
where the maximization is over random variables $\tilde{N}$ supported on $\mathcal{A}$ and satisfying the indicated constraint, $\phi_N (H(N))$ is readily shown to be  attained by $N$ (i.e.\ $\phi_N (H(N)) = H(N)$) uniquely (cf., e.g., \cite{ordentlichweissman}).  
Good lossy compression of the noisy source  $\mathbf{Z}$ under this distortion criterion at distortion level equal to the entropy of the noise ($D=H(N)$) turns out to result in reconstructions $\hX$ that are \emph{samples from the posterior} of the noise-free source $\bX$ given the noisy source $\bZ$. In particular, the finite-dimensional distributions of these reconstructions converge to those of the underlying noise-free source.   
 We state this phenomenon  rigorously in the theorem below. 
 ``Good code'' refers to a sequence of compressors, indexed by block-lengths, with respective rates and distortions converging to a point on the rate-distortion curve. 

  \begin{theorem}
  \label{theorem: emp dist of good codes for noisy data}
  Suppose $\mathbf{X}$ is a stationary ergodic process. Let $\{ \hX^n \}_{n \geq 1}$ be the reconstructions associated with 
a good code for the source $\mathbf{Z}$ with respect to the difference distortion function in (\ref{eq: distortion induced by noise}), at distortion level $H(N)$. 
For any finite $k$ and $n \geq k$, let $Q_{Z^k,\hX^k}^{\sf{ave}, (n)} = Q_{Z^k,\hX^k}^{\sf{ave}, (n)} [ P_{Z^n, \hX^n} ]$ and 
$Q_{Z^k,\hX^k}^{\sf{emp}, (n)}  = Q_{Z^k,\hX^k}^{\sf{emp}, (n)} [Z^n, \hX^n]$ denote, respectively, the $k$-th order joint distribution induced by $P_{Z^n, \hX^n}$ and the (random) $k$-th order joint distribution induced by the realized $(Z^n, \hX^n)$.  Then  \be Q_{Z^k,\hX^k}^{\sf{emp}, (n)}  \stackrel{n \rightarrow \infty}{\Longrightarrow}  P_{Z^k, X^k}
\ \ a.s. \ee and a fortiori 
   \be Q_{Z^k,\hX^k}^{\sf{ave}, (n)}  \stackrel{n \rightarrow \infty}{\Longrightarrow}  P_{Z^k, X^k}, \ee
where $P_{Z^k, X^k} $ is the joint $k$th-order distribution of the noisy and original noise-free source. 
\end{theorem} 
In particular, and most relevant for our purposes, the finite-dimensional distributions of \emph{lossy reconstructions of the noisy source} converge to those of the underlying \emph{noise-free source}.   \newline 

\begin{proof}
Let $R(Z^k, H(N))$ denote the $k$th-order rate-distortion function of $\mathbf{Z}$ at distortion level $H(N)$:
    \be
    R(Z^k, H(N)) = \min_{\E{[\frac{1}{k} \sum_{i=1}^k\rho_N(Z_i-\hat{X}_i)]} \leq H(N)} \frac{1}{k} I(Z^k; \hat{X}^k).
    \ee
Any pair $(Z^k, \hat{X}^k)$ within the feasible set would satisfy: 
\begin{align}
    \begin{aligned}
    I(Z^k;\hat{X}^k) & = H(Z^k) - H(Z^k|\hat{X}^k) \\
                    & = H(Z^k) - H(Z^k-\hat{X}^k|\hat{X}^k) \\
                    & \geq H(Z^k) - H(Z^k - \hat{X}^k) \\
                    & \geq H(Z^k) - \sum_{i=1}^k H(Z_i-\hat{X}_i) \\
                    & \stackrel{(a)}{\geq} H(Z^k) - \sum_{i=1}^k \phi_N \left( E \rho_N (Z_i-\hat{X}_i) \right) \\ 
                    & \stackrel{(b)}{\geq} H(Z^k) - k \phi_N \left( \frac{1}{k} \sum_{i=1}^k E \rho_N (Z_i-\hat{X}_i) \right) \\ & \stackrel{(c)}{\geq} H(Z^k) - k \phi_N \left( H(N) \right) \\
                    & = H(Z^k) - k H(N)
    \end{aligned}
    \label{eq_proof_4.1}
\end{align}
where (a) is by the definition of the function $\phi_N$, (b) by its (readily verified) convexity, (c) by its monotonicity along with the fact that the pair $(Z^k, \hat{X}^k)$ is in the feasible set,  and the last equality is due to the aforementioned property $\phi_N (H(N)) = H(N)$.    
 On the other hand, the pair $(Z^k, X^k)$ satisfies all the inequalities in (\ref{eq_proof_4.1}) with equality since $\mathbf{Z}$ is constructed as $Z_i = X_i + N_i$ and $N_i$s are IID$\sim N$, independent of $\mathbf{X}$. Furthermore, the invertibility of the Toeplitz matrix representing the PMF of $N$ (stipulated earlier) guarantees the uniqueness of the distribution  satisfying the feasibility condition and achieving $R(Z^k, H(N))$. Therefore, the $k$th-order rate-distortion function of $\mathbf{Z}$ with distortion level $H(N)$ is
\be
\label{eq: for of RZK at H(N)} 
R(Z^k, H(N)) = \frac{1}{k} H(Z^k) - H(N)
\ee
and it is uniquely achieved in distribution by the pair $(Z^k, X^k)$.  Taking $k \rightarrow \infty$ we also obtain 
\be
\label{eq: for of RZbold at H(N)} 
R( \mathbf{Z}, H(N)) =  \bar{H}(\mathbf{Z}) - H(N), 
\ee
where $\bar{H}(\mathbf{Z})$ is the entropy rate of $\mathbf{Z}$. Combining \eqref{eq: for of RZK at H(N)} and \eqref{eq: for of RZbold at H(N)} yields 
\be 
R(Z^k, H(N)) = 
R( \mathbf{Z}, H(N)) +  \frac{1}{k} H(Z^k) - 
\bar{H}(\mathbf{Z}). 
\ee
Thus, Part 2 and Part 3 of \cite[Theorem 9]{ordentlichweissman} are satisfied by the process $\mathbf{Z}$ for distortion measure $\rho_N$ at distortion level $H(N)$,   and we therefore have 
\be
 Q_{Z^k,\hX^k}^{\sf{emp}, (n)}  \stackrel{n \rightarrow \infty}{\Longrightarrow}  P_{Z^k, X^k}.
\ee
\end{proof}




\section{Application for Learning}
\label{application}


\subsection{Learning with Lossily Compressed Noisy (LCoN) Examples}
\label{LCON_framework}
Consider first the standard framework of unsupervised learning from $M$ non-labeled examples $\{ X^{n , (i)}   \}_{i=1}^M$, drawn IID $\sim X^n$.     The $i$th example comprises the  data point/signal/image  $X^{n , (i)}$,  which is an $n$-tuple with $\mathcal{A}$-valued components.
Our data pre-processing method,  illustrated in Fig.~\ref{fig:framework}, comprises noise injection and lossy compression  to obtain and store the lossily compressed noisy (LCoN) examples, as follows: 

\begin{enumerate}
    \item Pick a distribution for the noise $N$ (we discuss the choice of distribution later). Inject  IID$\sim N$ noise components to each component of each of the $X^{n , (i)}$s. Denote the noisy examples as $Z^{n , (i)}$, which are IID $\sim Z^n$, the noisy version of $X^n$. 
    
    \item Pick a good lossy compressor for the distortion function $d(z^n, \hat{x}^n)=\frac{1}{n} \sum_{i=1}^n \rho (z_i-\hat{x}_i) $ where $\rho(\cdot)$ is the distortion measure in (\ref{eq: distortion induced by noise}) and for distortion level equal to the entropy of the noise, i.e., $D
    =H(N)$. Jointly compress all the noisy data. Denote the reconstructions from the lossy compression of $Z^{n , (i)}$s as $ \hX^{n , (i)} $. 
    
    \item Use $ \hX^{n , (i)} $ instead of the $ X^{n , (i)} $ for learning. 
\end{enumerate}

Although the above describes jointly compressing all the data, 
one may also consider a more practical version where each example is compressed separately, as we elaborate below.     

\subsection{Data Efficiency while Retaining the Right Distribution}
\label{LCON_memory}
 
 What will be the cost of storing the compressed noisy data? Assuming the compressors employed are ``good'' in the sense of the previous section, it follows by invoking \cite[Theorem 4]{ordentlichweissman} 
 that, in the limit $M \rightarrow \infty$ of a large amount of training data, we will need a rate of 
 \be 
 \label{eq: rate needed}
 \frac{1}{n} H(Z^n) - H(N) \ \ \ \frac{\mbox{bits}}{\mbox{data component}},  
 \ee
 namely the rate distortion function of the IID $\sim Z^n$ source at distortion level $H(N)$, i.e., achieving the \emph{theoretically optimal compression performance}.
  Furthermore, Theorem~\ref{theorem: emp dist of good codes for noisy data} assures us that 
 $\{ \hX^{n , (i)}  \}_{ 1 \leq i \leq M  }$ will have an empirical distribution converging to the distribution of  $X^{n}$  
 when $M \rightarrow \infty$. Thus, overall,  the empirical distribution of 
  $\{ \hX^{n , (i)} \}_{i=1}^M$ converges in distribution to the right one, namely that of $X^n$.  
 Therefore, in the limit of many training examples, performing the learning on $\{ \hX^{n , (i)} \}_{i=1}^M$ should be as good as performing it on the original noise-free data  $\{ X^{n , (i)}   \}_{i=1}^M$. 
 
 
 We note that the foregoing discussion was valid for a fixed $n$ and an arbitrarily distributed $X^n$, in the $M \rightarrow \infty$ limit. It is also meaningful to consider a fixed $M$ in the large $n$ limit. Indeed,  when it is reasonable to think of the   generic  $X^n$ governing the data as the first $n$ components of a stationary ergodic process,  
 even if we merely employ good compressors separately on each example, 
 Theorem \ref{theorem: emp dist of good codes for noisy data} guarantees that the reconstructions will tend to be loyal to the original data in the sense of their finite-dimensional distributions, when $n$ is large. The assumption of a stationary ergodic process governing the examples may be natural in a variety of applications, such as when the $X^{n , (i)}$s represent audio signals or text.  Also, things carry over naturally to multi-dimensionally indexed data, e.g., when the $X^{m \times n , (i)}$s represent images sampled from the generic  $X^{m \times n}$, representing the  $m \times n$ grid of samples from a (spatially) stationary ergodic random field.

 \subsection{Privacy}
\label{LCON_privacy}
We would also like to guarantee that the database retained for the learning  does not leak too much information about any of the individual examples. To this end, we consider the (normalized) mutual information between the two, 
known as the \emph{privacy leakage},  which comes with a variety of operational justifications on top of its intuitive appeal (cf. \cite{makhdoumi2014information, du2012privacy, makhdoumi2013privacy}, references therein and thereto). 
 For each $i$, we have  \begin{align}
   \frac{1}{n} I \left( X^{n, (i)}; \left\{ \hX^{n, (j) } \right\}_{j=1}^M  \right)  & \stackrel{(a)}{\leq}    
   \frac{1}{n} I \left( X^{n, (i)}; \left\{ Z^{n, (j) } \right\}_{j=1}^M  \right) \nonumber \\
     & = \frac{1}{n} I \left( X^{n, (i)};  Z^{n, (i) }   \right) \\
     & = \frac{1}{n} I \left( X^{n};  Z^{n }   \right) \\
     & = \frac{1}{n} H \left(   Z^{n }   \right) - H(N) , \label{eq: upper bound on the privacy leakage}
\end{align}
where the inequality is due to data processing and the two equalities follow by  $\left( X^{n, (i)},  Z^{n, (i)} \right)$s being IID $\sim ( X^{n},  Z^{n }  )$. The inequality (a) will in general be quite loose  as the compression of the noisy examples is lossy, e.g., the lossy compressor makes the noisy image less recognizable in Fig.~\ref{fig:framework}. Tighter (and better) bounds on the privacy leakage could be attained when considering specific compressors. 

We now briefly discuss the privacy guarantees in terms of differential privacy -- a worst-case property as opposed to the average-case mutual information metric.  
\subsubsection{Differential Privacy}
With specific noise distributions such as Gaussian and Laplacian, we can quantify the \emph{local differential privacy} guarantees as well. For instance, in the experiments in Section~\ref{exp_gaussian}, we inject iid Gaussian noise to the images in the CelebA dataset. 
Assuming that each image in the CelebA dataset was provided by a different user (i.e. each user released a single image), the implemented scheme satisfies $(\epsilon, \delta)$-differential privacy when the variance of the injected Gaussian noise is $\sigma^2 = \frac{2 \log{ \left (2 / \delta \right )}}{\epsilon^2}$ \cite{dwork2006our} \footnote{We take the sensitivity as $\Delta = 1$.}. Similarly, in Section~\ref{exp_laplacian}, we inject iid Laplacian noise $\text{Lap}(x;b) = \frac{1}{2b} e^{-|x|/b}$ to the images and provide $\epsilon$-differential privacy when the parameter of the injected Laplacian noise is $ b = \frac{1}{\epsilon} $. If each user had released $k$ images in the CelebA dataset, then a slightly weaker differential privacy guarantee would be satisfied because each user would effectively release $k$ samples (or $k$ queries) from their sensitive data.\footnote{This is different from the group privacy where there are groups of ``correlated'' samples from \emph{different} individuals.} While each query individually meets the differential privacy guarantee given above, the composition of $k$ queries may leak more information and hence degrades the differential privacy guarantee. Under this scenario, by the composition theorem for differential privacy \cite{kairouz2015composition}, $(\epsilon, \delta)$-differential privacy is satisfied by corrupting the data with Gaussian noise with variance $\sigma^2 = \frac{8k \log{\left ( e + (\epsilon / \delta)\right )}}{\epsilon ^2}$ or with Laplacian noise with variance $  2b^2 = \frac{8k \log{\left ( e + (\epsilon / \delta)\right )}}{\epsilon ^2}$, where $k$ is the number of images released by each user.


 \subsection{Choice of the Noise Distribution}
\label{noise_details}
How should one choose the distribution of the noise? 
  The higher its entropy, the smaller the respective compression rate and upper bound on the privacy leakage in    
 (\ref{eq: rate needed}) and (\ref{eq: upper bound on the privacy leakage}) so, in principle, we get simultaneously  better compression \emph{and} more privacy. In fact, one could get both the compression rate and privacy leakage arbitrarily small with a noise distribution sufficiently close to uniform\footnote{Uniform itself is not allowed as per the stipulation of the noise distribution being non-singular.} since both (\ref{eq: rate needed}) and (\ref{eq: upper bound on the privacy leakage}) are upper bounded by  
 \be
 \log{|\mathcal{A}|} - H(N).
 \ee
 
 
 The choice of noise distribution, however, affects the convergence rate in large $n$ and  $M$ limits. 
 As a result, in practice, when both $n$ and $M$ are finite, there is a tension  between 
getting good (low) compression rate plus privacy leakage  and the quality (proximity to the true distribution) of the reconstructions. One might envision turning a knob sweeping through noise distributions to find a good sweet-spot. A more principled understanding of this point is left for future work.


 \subsection{Supervised Learning}  
 The foregoing framework and results carry over straightforwardly to the case 
 when the noise-free data come as $M$ \emph{labeled} examples $\{ ( X^{n , (i)} , L_i )  \}_{i=1}^M$, drawn IID $\sim (X^n, L)$, where   the labels $L_i$ take values in a finite alphabet of labels $\mathcal{L}$. 
 In this case, we apply the operations and arguments discussed above separately on each subset of the data pertaining to each label value.  The experimental results of Section~\ref{experiments} are in this setting. 
 
 \subsection{When Compression is Not Matched to the Noise}
 The following addresses many of the natural scenarios arising in practice where the lossy compression is tailored for a distortion function and/or level \emph{not matched} to the added noise characteristics. 
 
  \begin{corollary}
 Suppose the added noise  is  decomposable as     $N = U + W$, where $U$ and $W$ are independent. If a good code for the source $\bZ$ with respect to  $\rho_W$  at distortion level $H(W)$ is utilized then
 \be
 Q_{Z^k,\hX^k}^{\sf{emp}, (n)}  \stackrel{n \rightarrow \infty}{\Longrightarrow}  P_{Z^k, \tilde{X}^k}  \ \ a.s. 
 \ee
  and a fortiori 
  \be
  Q_{Z^k,\hX^k}^{\sf{ave}, (n)}  \stackrel{n \rightarrow \infty}{\Longrightarrow}  P_{Z^k, \tilde{X}^k},
  \ee
  where $\mathbf{\tilde{X}} = \mathbf{X} + \mathbf{U}$ (with $\mathbf{U}$ IID$\sim U$ and independent of $\mathbf{X}$)    
    is the partially noisy source and $P_{Z^k, \tilde{X}^k}$ is the joint kth-order distribution of the noisy and  partially noisy source. 
  \label{theorem: partial noise}    
   \end{corollary}
   
   \begin{proof}
    The proof follows from Theorem~\ref{theorem: emp dist of good codes for noisy data} by replacing $\mathbf{X}$ with $\mathbf{\tilde{X}} = \mathbf{X} + \mathbf{U}$ and $\mathbf{N}$ with $\mathbf{W}$. 
   \end{proof}
 
 Evidently, 
  in the scenarios covered by the theorem,  the lossy  compression denoises $\bZ$ only partially. For example, when applied to the case of added Gaussian noise and compression under squared error distortion, the theorem suggests that if compression is done under distortion $D_2$ smaller than the variance $\sigma^2$ of the noise then the reconstructions are effectively samples from the distribution of the noise-free data corrupted by Gaussian noise of variance $\sigma^2 - D_2$. Similarly, in the case of added Laplacian noise with distribution $\text{Lap}(x;b) = \frac{1}{2b} e^{-|x|/b}$ and compression under absolute error distortion, Corollary~\ref{theorem: partial noise} suggests that if compression is done under distortion $D_1$ smaller than $|b|$ then the reconstructions are samples from the distribution of the noise-free data corrupted by a noise with \emph{sparse Laplacian distribution} as follows:
  \be
  \text{SparLap}(x; b, D_1) = \frac{D_1^2}{b^2} \cdot \delta(x) + (1-\frac{D_1^2}{b^2}) \cdot \frac{1}{2b} e^{-|x|/b}.
  \ee
  Notice that the initially added Laplacian noise $N \sim \text{Lap}(x;b)$ can be decomposable as $N=U+W$, where $W \sim~\text{Lap}(x;D_1)$ and $U \sim~\text{SparLap}(x; b, D_1)$. While lossy compression under distortion $D_1$ removes the $W$ component, the data remains partially noisy due to $U$. The implications of this phenomenon for robustness will be explored in future work 
  (and briefly touched on experimentally in the next section).



%% file: sections/04-experiments.tex
\section{Experimental Results}
\label{experiments}


In this section, we test our suggested pipeline 
in the context of training a gender classifier on the CelebA dataset \cite{celeba}, consisting of 202,599 face images
of celebrities (cf.\ left image in Fig. \ref{fig:framework} for an example), using the ResNet-34 architecture \cite{he2016deep}. We chose the CelebA dataset since privacy of face images is an emerging concern, cf., e.g., a recent work \cite{yang2021study} studied the effect of face obfuscation in the context of the ImageNet challenge. In Sections~\ref{exp_gaussian} and~\ref{exp_laplacian}, we corrupt the original images in the CelebA dataset with (appropriately discretized) Gaussian and Laplacian noise, respectively.

 \subsection{Gaussian Noise}
 \label{exp_gaussian}
 In the experiments in this section, we inject Gaussian noise to the original images in the CelebA dataset. The induced distortion function in (\ref{eq: distortion induced by noise}) with distortion level being the entropy of the noise essentially boil down to squared error with distortion level being the variance of the added Gaussian noise. The  ``good'' lossy compressor we employ in the experiments, guided by our framework, is JPEG \cite{pennebaker1992jpeg}, which was (arguably approximately) designed with squared error in mind. We tune the compression level so that the squared error distortion approximately matches the variance of the injected noise. 
 
  \begin{figure*}[!h]
    \centering %
        \subfigure[Tested on Noise-free Examples. ]{\includegraphics[width=.47\textwidth]{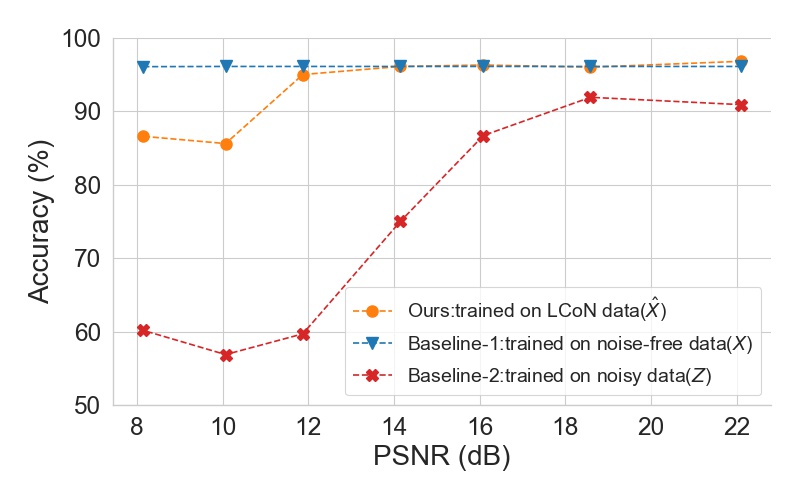}}
        \subfigure[Tested on Noisy Examples. ]{\includegraphics[width=.47\textwidth]{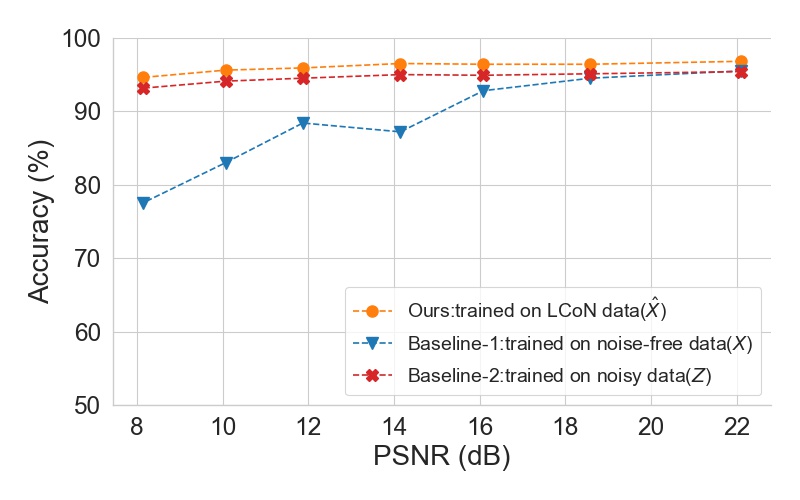}}
        \subfigure[Tested on LCoN Examples. ]{\includegraphics[width=.47\textwidth]{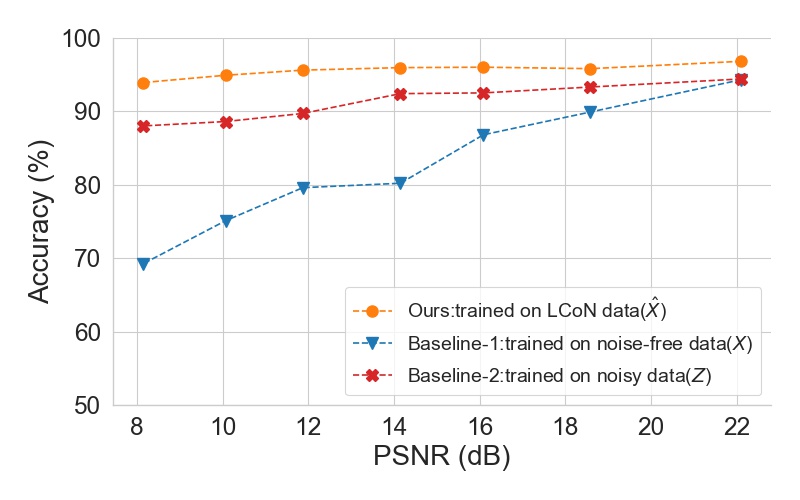}}
        \subfigure[Tested on Adversarial Examples. ]{\includegraphics[width=.47\textwidth]{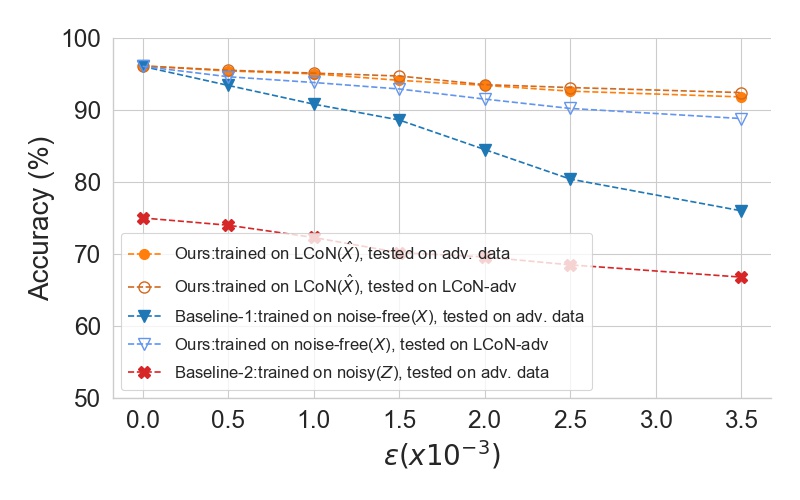}}
    \caption{Images corrupted with Gaussian noise. For the LCoN data, noisy images are compressed via JPEG. Comparison of models trained with LCoN examples, noise-free examples, and noisy examples on (a) noise-free test examples, (b) noisy test examples, (c) LCoN test examples, (d) adversarial test examples. Adversarial images are generated via the Fast Gradient Sign Method (FGSM) \cite{43405}. The compression rates of LCoN for PSNR $[8.1, 10.1, 11.9, 14.2, 16.1, 18.6, 22.1]$ in (a-c) are $[0.131, 0.136, 0.160, 0.170, 0.173, 0.179, 0.200]$, respectively.}\label{fig:comp_plot}
  
   \end{figure*}
   
  In Fig.~\ref{fig:comp_plot}, we compare three training schemes: 
  
\textbf{1) Our Setting - LCoN-train} (orange in Fig.~\ref{fig:comp_plot}): Training over reconstructions $\hX^{n, (i) }$ from \textbf{L}ossy \textbf{Co}mpression of \textbf{N}oisy examples. We call $\hX^{n, (i) }$s as LCoN-pre-processed examples. This setting comes with guarantees  on the privacy leakage 
    and storage cost of the data, as established in the previous section. Fig.~\ref{fig:framework} exhibits $X^{n, (i)}, Z^{n, (i)}$ and $\hX^{n, (i) }$ for a randomly chosen $i$ at the specified noise level and corresponding distortion.  
    
    
\textbf{2) Baseline-1} (blue in Fig.~\ref{fig:comp_plot}): Training over the noise-free examples $X^{n, (i)}$ from the CelebA dataset. 
    This method does not preserve privacy since the noise-free data are retained. 
    
\textbf{3)  Baseline-2} (red in Fig.~\ref{fig:comp_plot}): Training over noisy examples $Z^{n, (i)}$, injected with the same noise used in LCoN-train. 
    This time, the formal privacy guarantee is as good as LCoN-train's (although, in effect, as discussed, LCON-train provides better privacy due to the extra data processing step of compression). 

  After training, we test the respective three neural networks obtained (three for each noise level) on four different datasets: 

    \textbf{1) Noise-free test images} $X^{n, (i)}$ (Fig.~\ref{fig:comp_plot}(a)). 
    
  
    \textbf{2) Noise-injected test images} $Z^{n, (i)}$ -- with the same noise distribution used for the training data (Fig.~\ref{fig:comp_plot}(b)). 
    
    \textbf{3) LCoN-pre-processed test images} $\hX^{n, (i) }$ -- with the same noise and distortion used for the training data (Fig.~\ref{fig:comp_plot}(c)).

    \textbf{4) Adversarial test images} -- generated via the Fast Gradient Sign Method (FGSM) \cite{43405}  (Fig.~\ref{fig:comp_plot}(d)). 

In Fig.~\ref{fig:comp_plot}(a-c), PSNR refers to the PSNR of the noisy images (as dictated by the noise variance) after noise injection, prior to lossy compression. For a fair comparison, we calibrate the number of examples used by each scheme so that the overall storage cost (in bits) is approximately the same. In other words, in Fig.~\ref{fig:comp_plot}(a-c), the points on the same vertical line (same PSNR, same privacy) are trained with examples requiring the same storage cost by adjusting the number of training examples used. The compression rate for each point is provided in the caption of Fig.~\ref{fig:comp_plot}.  In Fig.~\ref{fig:comp_plot}(d), we vary the parameter $\epsilon$ in FGSM. Recall that FGSM corrupts the data as $x_{\text{adv}} = x + \epsilon \cdot \text{sign}(\nabla_x J)$, where $J$ is the loss function of the downstream task, i.e., the higher $\epsilon$ the more corrupted the adversarial data. In Fig.~\ref{fig:comp_plot}(d), in addition to testing directly on the adversarial data, we test LCoN-train and Baseline-1 on  LCoN-pre-processed adversarial data as well. We denote the pre-processed adversarial data as LCoN-adv (empty markers).



We observe that LCoN-train consistently outperforms Baseline-2 in all settings and noise levels. The gap is most significant when the models are tested on the noise-free images (Fig.~\ref{fig:comp_plot}(a)). This behavior is expected in light of  the theory exposed in the previous sections: LCoN examples $\hX^{n, (i) }$ are close in distribution to the noise-free examples $X^{n, (i)}$
so a model trained on LCoN examples should be expected to outperform one trained on the noisy ones $Z^{n, (i)}$. Perhaps less expected is that LCoN-train outperforms Baseline-2 even on the noisy data on which the latter was trained. The comparison to Baseline-1 is also extremely favorable (on top of the fact that Baseline-1 preserves no privacy) essentially across the deck. 
Even on the noise-free test data, our method yields essentially the same accuracy as Baseline-1 at sufficiently high PSNR. Remarkably,  our setting reaches $96.8\%$ accuracy for PSNR higher than $20$ dB, which is even higher than the $96.6 \%$ accuracy of the model trained with full noise-free CelebA dataset (not a subset to comply with the storage constraint, as in Baseline-1). Evidently, even when storage is free and privacy is not an issue, LCoN-train is an accuracy booster.   
Finally, Fig.~\ref{fig:comp_plot}(d) shows that LCoN-train is a significant performance booster in the face of adversarially corrupted data as well.  The gap between LCoN-train and the better of the other two benchmarks becomes as large as $16.4\%$ in accuracy. Furthermore, even if the model is trained on the noise-free data,  LCoN pre-processing of the adversarial testing data can result in as much as $13 \%$ of an accuracy boost.  Overall, it seems, LCoN pre-processing is advisable both at training and testing (and at just one of them if the other is fixed).  

\begin{figure}[h]
    \centering %
        \includegraphics[width=\columnwidth]{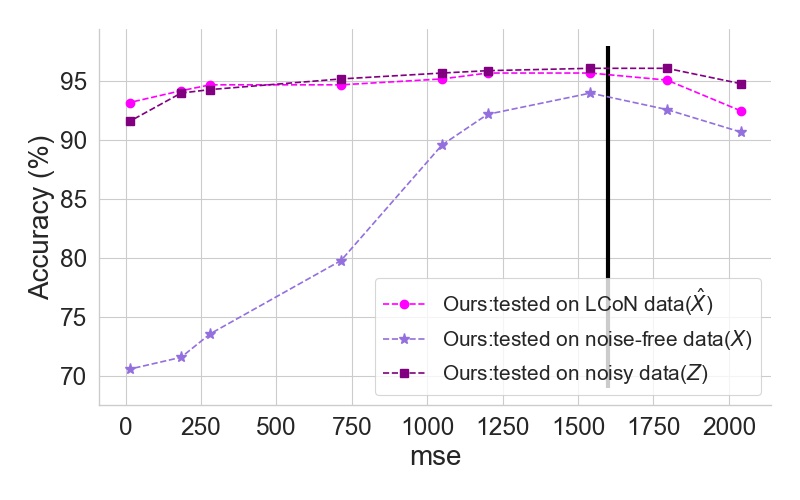}
    \caption{LCoN-train with training data added Gaussian noise $N(0, 1600)$ and compressed with varying distortion levels (mse). Compression is done via JPEG. Vertical line corresponds to mse=$1600$, where the distortion level and the entropy of the noise are matched.}\label{fig:varying_dist_level_plot}
   \end{figure}

Finally, Fig.~\ref{fig:varying_dist_level_plot} shows that the best accuracy (across all test data sets) is obtained when the distortion level (mse) is closest to the entropy of the noise. For all the points in Fig.~\ref{fig:varying_dist_level_plot}, a $N(0, 1600)$ Gaussian noise is added to the training data. The black vertical line corresponds to mse~=$1600$ where the distortion level is matched to the entropy of the noise. As expected from Corollary~\ref{theorem: partial noise}, when the distortion level decreases, the lossily reconstructed examples $\tilde{X}^{n, (i) }$ the model is trained on become more noisy. This results in a model trained on examples with distribution further away from the distribution of the $X^{n, (i) }$s,  explaining the significant accuracy drop, especially on the noise-free test data, when $D<H(N)$. A similar effect occurs  when $D>H(N)$. 

 \subsection{Laplacian Noise}
 \label{exp_laplacian}
   \begin{figure*}[!h]
    \centering %
        \subfigure[Tested on Noise-free Examples. ]{\includegraphics[width=.47\textwidth]{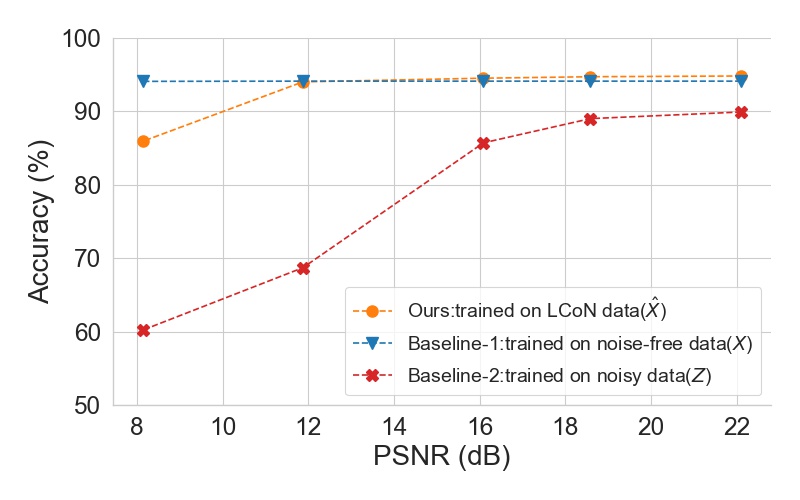}}
        \subfigure[Tested on Noisy Examples. ]{\includegraphics[width=.47\textwidth]{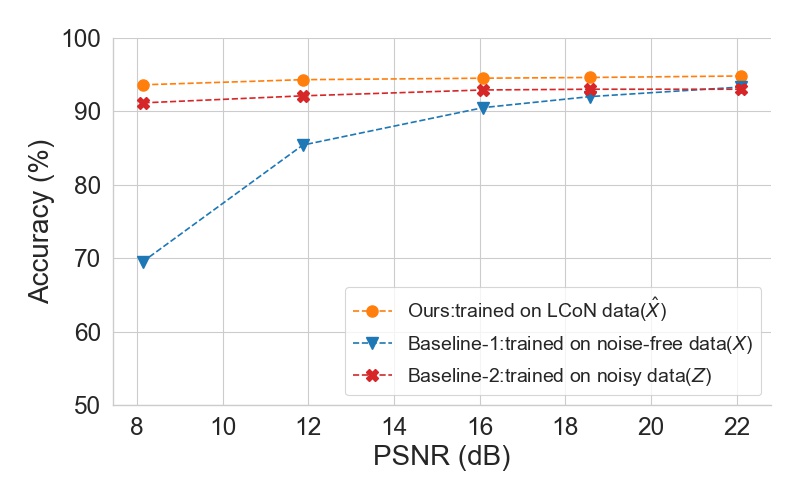}}
        \subfigure[Tested on LCoN Examples. ]{\includegraphics[width=.47\textwidth]{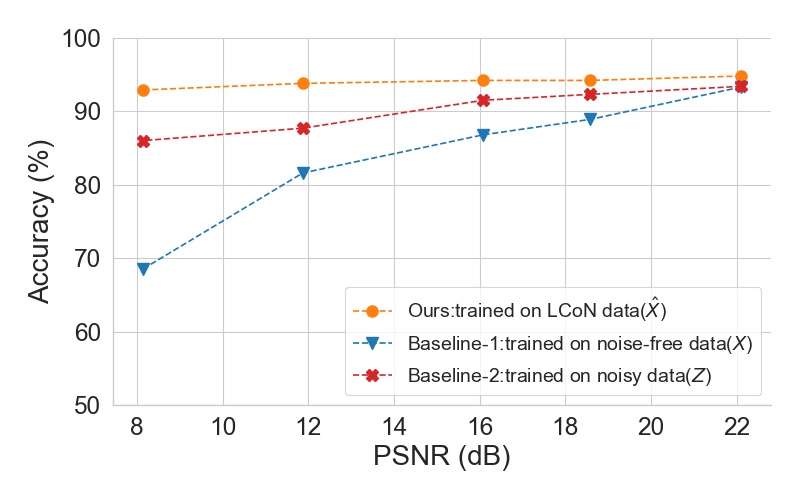}}
        \subfigure[Tested on Adversarial Examples. ]{\includegraphics[width=.47\textwidth]{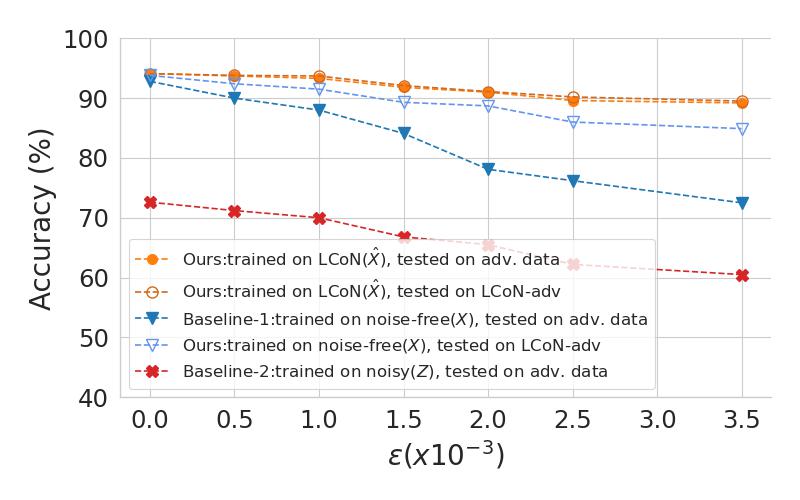}}
    \caption{Images corrupted with Laplacian noise. For the LCoN data, noisy images are compressed via a neural compressor trained to minimize mean absolute error (mae, or $\ell_1$ loss). Comparison of models trained with LCoN examples, noise-free examples, and noisy examples on (a) noise-free test examples, (b) noisy test examples, (c) LCoN test examples, (d) adversarial test examples. Adversarial images are generated via the Fast Gradient Sign Method (FGSM) \cite{43405}. The compression rates of LCoN for PSNR $[8.1, 11.9, 16.1, 18.6, 22.1]$ in (a-c) are $[0.150,  0.172, 0.190, 0.199, 0.212]$, respectively.} \label{fig:comp_plot_lap}
    \end{figure*}
 
 In this section, we show results similar to those in Section~\ref{exp_gaussian} using Laplacian noise instead of Gaussian and a neural image compressor instead of JPEG. Note that, with noise distribution $\text{Lap}(x;b) = \frac{1}{2b} e^{-|x|/b}$, the induced distortion function in (\ref{eq: distortion induced by noise}) with distortion level being the entropy of the noise corresponds to mean absolute error (mae) with the distortion level being $b$. Since we need a compressor that optimizes for mae, we first train a neural image compressor using a subset of the CelebA dataset. Specifically, we follow the end-to-end training approach proposed in \cite{BaLaSi17, Ba18, BaMiSiHwJo18, minnen2018joint} and train a variational autoencoder to minimize a Lagrangian cost function $L(\lambda) = D + \lambda R$, where $D$ is the mae between the original and the predicted image and $R$ is the estimated bit rate using a continuous relaxation of the probability model. We use a version of the open-sourced PyTorch implementation \cite{NDIC} of \cite{BaLaSi17}. By varying the Lagrangian parameter $\lambda$, we train several neural compressors that would generate reconstructions with different distortion levels. Each distortion level matches a particular noise level used in the noise injection step. All in all, for each noise level we try, we use the corresponding trained neural compressor with the right distortion level to compress the noisy images. 
 
 Fig.~\ref{fig:comp_plot_lap} shows the results across various noise levels and verifies that Laplacian noise injection followed by an end-to-end trained neural image compression (trained to minimize mae) provides similar compression, privacy, accuracy, and robustness boosts that we achieve with the Gaussian noise injection followed by JPEG compression in Section~\ref{exp_gaussian}. Notice that the accuracy levels are slightly lower than the ones in Fig.~\ref{fig:comp_plot}. This is because we separate a portion of the CelebA dataset to train the neural compressor and not use it while training the classifier for a meaningful analysis. 

%% file: sections/05-conclusion.tex
\section{Conclusion and Future Work}
\label{conclusion}
Guided by and combining existing theory on lossy noisy data  
compression and on information-theoretic privacy, we proposed a data pre-processing procedure for both training and testing data which appears to simultaneously boost  data efficiency, privacy, accuracy and robustness. Our theoretical framework has accounted for much of the empirical observations as they pertain to the efficiency (compression), privacy (leakage) and accuracy (due to preservation of the right distribution). The robustness is a welcome additional feature we have observed empirically, and perhaps to be intuitively expected given empirical work showing that noise injection \cite{45818} and image compression \cite{robustness_jpeg1, robustness_jpeg3, robustness_jpeg5}, when applied to adversarial data (each separately), improves robustness. Future work will be dedicated to quantifying this effect via (an extension of) our theoretical framework. From a high-level perspective, LCoN and dithered quantization  \cite{gray1993dithered} have some resemblance in injecting noise prior to compression. However, we employ noise injection, independent of the lossy compression, to preserve privacy while the added noise in dithered quantization is an essential component of the quantization step. We also note that our framework and theoretical insights transfer directly to the case where the data are noise-corrupted to begin with (rather than the noise being deliberately injected).  In such a case, the compression would be tuned to the real noise characteristics. Practically, we plan to further the experiments to other noise distributions and compressors such as  PNG \cite{PNG}, JPEG XR \cite{jpegXR}, WebP \cite{WebP}, sandwiched image compressor \cite{guleryuz2021sandwiched, guleryuz2022sandwiched}, sandwiched video compressor \cite{isik2023sandwiched}, SuRP \cite{isik2022information, isik2021successive}, LVAC \cite{isik2021lvac}, and LFZip \cite{LFZip2020}, which would be equally natural to experiment with, so long as they are appropriately matched (Gaussian noise for compressors designed with squared error in mind, Laplacian noise for compressors optimized for absolute error such as SuRP \cite{isik2022information, isik2021successive}, Uniform distribution on a sub-interval of length equal to the allowed maximum distortion for compressors designed under a maximal distortion criterion such as LFZip \cite{LFZip2020}, etc.). 
 Better compressors will likely boost the performance under the other criteria as well. 
 

 

%% file: sections/06-acknowledgement.tex
\section{Acknowledgement}
The authors would like to thank Shubham Chandak, Dmitri Pavlichin, Peter Kairouz, and Wei-Ning Chen for helpful discussions and the anonymous reviewers for valuable feedback. This work was supported by a Stanford Graduate Fellowship, a National Science Foundation (NSF) award, and Siemens and Meta research awards.